\documentclass[nofootinbib,twocolumn,pra,superscriptaddress]{revtex4-2}
\usepackage{hyperref}
\usepackage{latexsym,amsmath,amssymb,amsfonts,graphicx,color,amsthm}
\usepackage{enumerate}
\usepackage{braket}
\usepackage{adjustbox}
\usepackage{footnote}
\usepackage[dvipsnames]{xcolor}
\usepackage{tikz}
\usepackage{subcaption}
\usepackage{bbm}
\usepackage{verbatim}

\newcommand{\cH}{\mathcal{H}}
\newcommand{\cI}{\mathcal{I}}
\newcommand{\cJ}{\mathcal{J}}
\newcommand{\cK}{\mathcal{K}}
\newcommand{\cL}{\mathcal{L}}

\newcommand{\cQ}{\mathcal{Q}}
\newcommand{\cR}{\mathcal{R}}
\newcommand{\cS}{\mathcal{S}}
\newcommand{\cT}{\mathcal{T}}

\newcommand{\tA}{\tilde{A}}
\newcommand{\tPhi}{\tilde{\Phi}}

\newcommand{\tI}{\tilde{\mathcal{I}}}
\newcommand{\tK}{\tilde{\mathcal{K}}}
\newcommand{\tH}{\tilde{\mathcal{H}}}
\newcommand{\tx}{\tilde{x}}
\newcommand{\ty}{\tilde{y}}

\def \diracspacing {0.7pt}
\newcommand{\ketbra}[2]{| \hspace{\diracspacing} #1 \rangle \langle #2 \hspace{\diracspacing} |} 
\newcommand{\ketbraq}[1]{\ketbra{#1}{#1}} 

\newcommand{\Id}{\mathbb{I}}

\newcommand{\IdC}{\mathrm{id}}

\newcommand{\Ins}{\mathbbm{I\hspace{-0.1em}n}}
\newcommand{\M}{\mathbbm{M}}
\newcommand{\Ch}{\mathbbm{C\hspace{-0.1em}h}}

\DeclareMathOperator{\tr}{tr}

\newtheorem{theorem}{Theorem}
\newtheorem{proposition}[theorem]{Proposition}
\newtheorem{lemma}[theorem]{Lemma}
\newtheorem{corollary}[theorem]{Corollary}
\newtheorem{definition}[theorem]{Definition}

\newtheorem{remark}[theorem]{Remark}

\begin{document}

\title{Characterizing and quantifying the incompatibility of quantum instruments}
\author{Arindam Mitra}
\email{amitra@imsc.res.in}
\affiliation{Optics and Quantum Information Group, The Institute of Mathematical Sciences,
C. I. T. Campus, Taramani, Chennai 600113, India}
\affiliation{Homi Bhabha National Institute, Training School Complex, Anushaktinagar, Mumbai 400094, India}

\author{M{\'a}t{\'e} Farkas}
\email{mate.farkas@icfo.eu}
\affiliation{ICFO-Institut de Ciencies Fotoniques, The Barcelona Institute of Science and Technology, 08860 Castelldefels, Spain}

\date{\today}

\begin{abstract}
Incompatibility of quantum devices is one of the cornerstones of quantum theory, and the incompatibility of quantum measurements and channels has been linked to quantum advantage in certain information theoretic tasks. In this work, we focus on the less well-explored question of the incompatibility of quantum instruments, that is, devices that describe the measurement process in its entirety, accounting for both the classical measurement outcome and the quantum post-measurement state. In particular, we focus on the recently introduced notion of parallel compatibility of instruments, which has been argued to be a natural notion of instrument compatibility. We introduce---similarly to the case of measurements and channels---the incompatibility robustness of quantum instruments and derive universal bounds on it. We then prove that post-processing of quantum instruments is a free operation for parallel compatibility. Last, we provide families of instruments for which our bounds are tight, and families of compatible indecomposable instruments.
\end{abstract}

\maketitle

\section{Introduction}

The incompatibility of devices is a key feature of quantum theory that sets it apart from classical physics \cite{HMZ16}. The notion of incompatibility captures the fact that some devices (such as measurements or state transformations) cannot be performed simultaneously. While at first incompatibility may sound like a drawback, it turns out that incompatibility is an essential resource in various information processing tasks with quantum advantage \cite{Kim-Bell-incomp,Vertesi-steer-incomp,Skrzypczyk-state-dis-incomp,Heino-QRAC}. Furthermore, incompatibility is linked to fundamental concepts such as Bell nonlocality \cite{WPF09}, Einstein--Podolsky--Rosen steering \cite{Vertesi-steer-incomp,UBGP15} or contextuality \cite{TU20}. For this reason, characterizing incompatible devices and quantifying incompatibility is a fundamental research program relevant to quantum information processing.

Arguably, the type of quantum devices whose incompatibility is the most studied is that of quantum measurements, while there also exist various results regarding the incompatibility of quantum channels (i.e., state transformations). The third commonly studied family of quantum devices generalizes measurements and channels and is called quantum instruments. Instruments take a quantum state as an input, and produce a classical and a quantum output. That is, quantum instruments describe the full measurement process including the measurement outcome (the classical output) and the post-measurement state (the quantum output).

While the compatibility of quantum instruments has been studied in the literature, it has recently been pointed out that the ``traditional'' notion of instrument compatibility might not be satisfactory in some cases \cite{MF22}. The authors of Ref.~\cite{MF22} introduce the notion of \textit{parallel} compatibility and argue that it is a more natural notion of instrument compatibility than traditional compatibility. In this work we initiate the characterization of parallel compatible instruments and the quantification of parallel incompatibility.

We start with a few technical preliminaries in Section~\ref{Sec:preli} on quantum devices, their compatibility, and quantifying incompatibility. Then in Section~\ref{sec:main-charac-comp-ins} we introduce the incompatibility robustness of quantum instruments, prove bounds on this quantity, prove that post-processing of instruments is a free operation for parallel compatibility, and provide families of instruments for which our bounds are tight, as well as families of indecomposable compatible instruments. Finally, in Sec. \ref{Sec:conc}, we summarize our results and discuss future directions.

\section{Preliminaries}\label{Sec:preli}
Every quantum system has an associated Hilbert space $\cH$ that (while most definitions and statements in this paper generalize to the infinite dimensional case) we will assume to be finite dimensional. We denote the set of linear operators on the Hilbert space $\cH$ by $\cL(\cH)$ and the set of positive semidefinite linear operators on the Hilbert space $\cH$ by $\cL^+(\cH)$. We denote the set of quantum states (i.e., density matrices, or positive semidefinite operators with trace one) on $\cH$ by $\cS(\cH)$.

\subsection{Quantum measurements}
A quantum measurement $A$ with a finite outcome set $\Omega_A$ is described by a set of $| \Omega_A |$ positive semidefinite operators on $\cH$, i.e., $A=\{A(x)\}_{x \in \Omega_A}$ such that $\sum_{x \in \Omega_A} A(x) =\Id_{\cH}$ where $\Id_{\cH}$ is the identity operator on~$\cH$ and $| \Omega_A |$ is the cardinality of $\Omega_A$ \cite{Heino-book}. If the measurement $A$ is performed on a quantum state $\rho\in\cS(\cH)$, the probability to obtain the outcome $x$ is given by $\tr[\rho A(x)]$. We denote the set of measurements with outcome set $\Omega$ acting on the Hilbert space $\cH$ by $\M(\Omega,\cH)$. A measurement $A$ is said to be a trivial measurement if $A(x)=p_x\Id$ for all $x\in\Omega_A$ for some probability distribution $p_x$. In general, we will use the term positive operator-valued measure (POVM) interchangeably with the term quantum measurement. Furthermore, if $A^2(x)=A(x)$ for all $x\in\Omega_A$, then the measurement is called \textit{projective}, or \textit{projection-valued-measure} (PVM).

A pair of measurements $(A,B)$ on $\cH$ is said to be \textit{compatible} \cite{HMZ16} if there exists a \textit{joint measurement} $G$ on $\cH$ with outcome set $\Omega_A\times\Omega_B$ such that
\begin{align}
A(x)=\sum_{y\in\Omega_B}G(x,y) \quad \forall x \in \Omega_A\\
B(y)=\sum_{x\in\Omega_A}G(x,y) \quad \forall y \in \Omega_B
\end{align}
The outcome $(x,y)$ of the measurement $G$ on a given quantum state $\rho$ is distributed according to the joint distribution of the outcome of $A$ and $B$ on the same quantum state. Therefore, implementing $G$ amounts to simultaneously implementing $A$ and $B$.
Measurement pairs that are not compatible are called \textit{incompatible}, and the definition of (in)compatibility naturally generalizes to sets of more than two measurements. Measurement incompatibility has been linked to quantum advantage in various information processing tasks \cite{Kim-Bell-incomp,Vertesi-steer-incomp,Skrzypczyk-state-dis-incomp,Heino-QRAC}. Therefore, measurement incompatibility can be thought of as a resource in these tasks \cite{Buscemi-res-incomp}.

\subsection{Quantum operations and quantum channels}

A quantum operation $\Phi:\cS(\cH)\rightarrow\cL^+(\cK)$ is a completely positive (CP) trace non-icreasing map such that for any set of states $\{\rho_i\in\cS(\cH)\}$ and any probability distribution $\{p_i\}$, we have that $\Phi(\sum_ip_i\rho_i)=\sum_ip_i\Phi_i(\rho_i)$. It is known that every quantum operation has a unique linear extension $\tilde{\Phi}:\cL(\cH)\rightarrow\cL(\cK)$,
 i.e., a unique linear map $\tilde{\Phi}$ such that $\Phi=\tilde{\Phi}|_{\cS(\cH)}$ where $\tilde{\Phi}|_{\cS(\cH)}$ is the restriction of $\tilde{\Phi}$ to $\cS(\cH)$ \cite{Heino-book}.
 
In words, a quantum operation maps every quantum state to a sub-normalized quantum state. Quantum operations that map quantum states to normalized quantum states (and hence can be thought of as state transformations) are called \textit{quantum channels}. Formally, a quantum channel $\Lambda:\cS(\cH)\rightarrow\cS(\cK)$ is a trace preserving (TP) quantum operation, also called a CPTP map. We denote the set of channels with input space $\cH$ and output space $\cK$ by $\Ch(\cH, \cK)$.
 
Two quantum channels $\Lambda_1:\cS(\cH)\rightarrow\cS(\cK_1)$ and $\Lambda_2:\cS(\cH)\rightarrow\cS(\cK_2)$ are said to be \textit{compatible}  \cite{HMZ16} if there exists a \textit{joint channel} $\Lambda:\cS(\cH)\rightarrow\cS(\cK_1\otimes\cK_2)$ such that for all $\rho\in\cS(\cH)$ we have that
 \begin{align}
 \Lambda_1(\rho)=\tr_{\cK_2}\Lambda(\rho)\\
 \Lambda_2(\rho)=\tr_{\cK_1}\Lambda(\rho).
 \end{align}
The output of the joint channel $\Lambda$ is a joint (composite) state of the outputs of $\Lambda_1$ and $\Lambda_2$ on a tensor product Hilbert space. Channels that are not compatible are called  \textit{incompatible}, and the definition of (in)compatibility naturally generalizes to sets of more than two channels.
 The incompatibility of quantum channels has been studied in various works \cite{Heino-incom-chan,Plavala-incomp-chan,Plavala-incomp-chan-gen-prob} and it has been shown to provide advantage in quantum state discrimination \cite{Mori-incomp-chan}.

\subsection{Quantum Instruments}
A quantum instrument $\cI$ is described by a set of quantum operations, in particular, $\cI=\{\Phi_x:\cS(\cH)\rightarrow\cL^+(\cK)\}_{x\in\Omega_{\cI}}$ such that $\Phi=\sum_{x \in \Omega_\cI}\Phi_x$ is a quantum channel \cite{Davies-in,Heino-book,Pellonpaa-1,Pellonpaa-2}.
We denote the set of quantum instruments with outcome set $\Omega_{\cI}$, input space $\cH$ and output space $\cK$ by $\Ins(\Omega_{\cI},\cH,\cK)$.
Every quantum instrument induces a unique quantum measurement (a POVM)~$A$. That is, a unique set of positive semidefinite operators $A(x)$ on $\cH$ such that $\tr[ \rho A(x) ] = \tr[ \Phi_x(\rho) ]$ for all $\rho \in \cS(\cH)$ and all $x \in \Omega_\cI$. It is in fact straightforward to see that these operators are defined by $A(x) = \Phi^\ast_x(\Id)$, where $\Phi_x^\ast: \cL(\cK) \to \cL(\cH)$ is the dual of $\Phi_x$ defined via $\tr[ \Phi_x(\rho) M ] = \tr[ \rho \Phi^\ast_x(M) ]$ for all $\rho\in\cS(\cH)$ and for all $M \in \cL(\cK)$ . The probability of observing outcome $x$ upon measuring $A$ on $\rho$ is given by $\tr[ \Phi_x(\rho) ]$ and the (un-normalized) post-measurement state is given by $\Phi_x(\rho)$. An instrument inducing the POVM $A$ is sometimes called an $A$-compatible instrument. While every instrument induces a unique POVM, for a given POVM there exist many different instruments implementing it. Therefore, one can think of instruments as various ways of implementing a POVM, with various post-measurement states.

Similarly to the case of quantum measurements and quantum channels, one might call two quantum instruments \textit{compatible} if they can be performed simultaneously. The following definition of \textit{parallel compatibility} of quantum instruments is one way to make this intuition rigorous \cite{MF22}.

\begin{definition}[Parallel compatibility]\label{Def.parallel_comp}
A pair of quantum instruments $\cI_1=\{\Phi^{1}_x:\cS(\cH)\rightarrow\cL^+(\cK_1)\}_{x\in\Omega_{\cI_1}}$ and $\cI_2=\{\Phi^2_y:\cS(\cH)\rightarrow\cL^+(\cK_2)\}_{y\in\Omega_{\cI_2}}$ are (parallel) compatible if there exists a joint quantum instrument  $\cI=\{\Phi_{xy}:\cS(\cH)\rightarrow\cL^+(\cK_1\otimes\cK_2)\}_{x\in\Omega_{\cI_1},y\in\Omega_{\cI_2}}$ such that
\begin{align}
\sum_y\tr_{\cK_2}\Phi_{xy}=\Phi^{1}_x \quad \forall x \label{Eq:def-par-comp-1}\\
\sum_x\tr_{\cK_1}\Phi_{xy}=\Phi^{2}_y \quad \forall y \label{Eq:def-par-comp-2}.
\end{align}
\end{definition}

In words, the joint instrument $\cI$ simultaneously reproduces both the classical ($x$ and $y$) and the quantum [$\Phi^1_x(\rho)$ and $\Phi^2_y(\rho)$] outputs of $\cI_1$ and $\cI_2$. It has been argued in Ref.~\cite{MF22} that parallel compatibility is in some cases a favourable definition of instrument compatibility compared to the ``traditional'' definition \cite{HMR14,Lev16,Gud20}, which requires the existence of a joint instrument such that $\sum_y \Phi_{xy} = \Phi^1_x$ and $\sum_x \Phi_{xy} = \Phi^2_y$ for all $x$ and $y$. Therefore, in the following by ``compatible instruments'' we mean parallel compatible instruments, unless otherwise specified. It is also clear that this definition naturally generalizes to more than two instruments.

We now recall the notion of \textit{post-processing} of quantum instruments, which will be useful for characterizing compatible instruments \cite{LS21}.

 \begin{definition}
 Consider the instruments $\cI_1=\{\Phi^1_{x}\}_{x\in\Omega_{\cI_1}}\in \Ins({\Omega}_{\cI_1},\cH,\cK_1)$ and $\cI_2=\{\Phi^2_{y}\}_{y\in\Omega_{\cI_2}}\in \Ins({\Omega}_{\cI_2},\cH,\cK_2)$. Then, $\cI_2$ is said to be the post-processing of $\cI_2$ if there exists a set of instruments $\{\cR^x=\{R^{x}_y\}_{y\in\Omega_{\cI_2}}\in \Ins ({\Omega}_{\cI_2},\cK_1,\cK_2)\}_{x\in\Omega_{\cI_1}}$ such that
\begin{align}\label{eq:def-post-proc-in}
\Phi^2_{y}=\sum_{x\in\Omega_{\cI_1}} R^{x}_y\circ\Phi^1_{x} \quad \forall y \in \Omega_{\cI_2}.
\end{align} 
We denote this relation as $\cI_2 \lesssim \cI_1$, reflecting the fact that this relation induces a preorder on the set of instruments.
 \end{definition}
 
The post-processing preorder induces an equivalence relation: $\cI_1 \sim \cI_2$ if both $\cI_1 \lesssim \cI_2$ and $\cI_2 \lesssim \cI_1$ hold. On the equivalence classes post-processing becomes a partial order \cite{LS21}.

Two simple classes of instruments, defined directly through their induced POVMs, are so-called measure-and-prepare instruments. One can intuitively think of them as performing a measurement and then preparing a state depending on the outcome of the measurement.
 
\begin{definition}\label{Def:spec-meas-prep}
For a measurement $A \in \M( \Omega_A, \cH)$, a special measure-and-prepare instrument is defined as $\mathfrak{J}_{A}=\{J^A_x:\cS(\cH)\rightarrow\cL^+(\cK)\}_{x\in\Omega_A}$ with $\dim(\cK)= |\Omega_A|$, where
\begin{align}
J^A_x(\rho):=\tr[\rho A(x)]\ket{x}\bra{x}
\end{align}
 for all $\rho\in\cS(\cH)$ and $x\in\Omega_A$, and $\{\ket{x}\}_{x\in\Omega_A}$ is an orthonormal basis on $\cK$. The corresponding (special measure-and-prepare) quantum channel $\Gamma_A$ acting on a quantum state $\rho\in\cS(\cH)$ is given by $\Gamma_{A}(\rho)=\sum_xJ^A_x(\rho)=\sum_x\tr[\rho A(x)]\ket{x}\bra{x}$. 
 \end{definition}
 
 \begin{definition}\label{Def:meas-prep}
For a measurement $A\in \M( \Omega_A, \cH)$, a measure-and-prepare instrument is defined as $\mathfrak{J^{\prime}}_{A}=\{J^{\prime A}_x:\cS(\cH)\rightarrow\cL^+(\cK)\}_{x\in\Omega_A}$  such that for all $\rho\in\cS(\cH)$ and $x\in\Omega_A$
\begin{align}
J^{\prime A}_x(\rho):=\tr[\rho A(x)]\rho^{\prime}_x
\end{align}
where $\rho^{\prime}_x\in\cS(\cK)$ are some fixed quantum states. The corresponding (measure-and-prepare) quantum channel $\Gamma^{\prime}_A$ acting on a quantum state $\rho\in\cS(\cH)$ is given by $\Gamma^{\prime}_{A}(\rho)=\sum_xJ^{\prime A}_x(\rho)=\sum_x\tr[\rho A(x)]\rho^{\prime}_x$.
 \end{definition}
 
Last, measure-and-prepare instruments corresponding to trivial measurements are called \textit{trash-and-prepare} instruments, since their outcome does not depend on the input state.
 
  \begin{definition}\label{Def:trash-prep}
A measure-and-prepare instrument $\mathfrak{J^{\prime}}_{T}$ is called as a trash-and-prepare instrument if $T$ is a trivial measurement.
 \end{definition}

\subsection{Incompatibility robustnesses}

A common way of quantifying the incompatibility of a set of quantum devices is by asking how much noise needs to be added to the set in order to make it compatible \cite{HMZ16}. While this quantifier depends on what is considered to be ``noise'' \cite{DFK19}, in what follows we adopt the definition usually referred to as \textit{generalized robustness}.

For a pair of measurements $\{A_1(x)\}_{x \in \Omega_{A_1}}$ and $\{A_2(y)\}_{y \in \Omega_{A_2}}$ on $\cH$, the generalized incompatibility robustness is defined as
\begin{equation}\label{Def:ROI-measurements}
\begin{split}
R_M(A_1,A_2) = \min & \left. r \right. \\
\text{s.t.} ~ & \left. \frac{A_1(x) + r \tilde{A}_1(x)}{1+r} = \sum_{y} G(x,y) \quad \forall x \right. \\
& \left. \frac{A_2(y) + r \tilde{A}_2(y)}{1+r} = \sum_x G(x,y) \quad \forall y \right. \\
& \left. \tilde{A}_1(x), \tilde{A}_2(y) \ge 0 \quad \forall x,y \right. \\
& \left. \sum_x \tilde{A}_1(x) = \sum_y \tilde{A}_2(y) = \Id \right. \\
& \left. G(x,y) \ge 0 \quad \forall x,y. \right.
\end{split}
\end{equation}
In words, the generalized incompatibility robustness quantifies how much noise $(\tilde{A}_1,\tilde{A}_2)$ can be added to the pair $(A_1,A_2)$ before the noisy POVM pair $(A_1+r\tilde{A}_1)/(1+r)$ and $(A_2+r\tilde{A}_2)/(1+r)$ becomes compatible (with some joint measurement $G$). The term ``generalized'' refers to the fact that $\tA_1$ and $\tA_2$ can be chosen to be arbitrary POVMs on $\cH$ with outcome numbers matching with that of $A_1$ and $A_2$, respectively. Note that the incompatibility robustness can be cast as an efficiently computable semidefinite program (SDP) \cite{BV04} and that this definition naturally generalizes to more than two measurements.

The incompatibility robustness of two quantum channels $\Phi_1: \cS(\cH) \to \cL^+(\cK_1)$ and $\Phi_2: \cS(\cH) \to \cL^+(\cK_2)$ can be defined similarly:
\begin{equation}\label{Def:ROI-channels}
\begin{split}
R_{C}(\Phi_1,\Phi_2)=\min~  & \left. r\right. \\
\text{s.t.} ~ & \left. \frac{\Phi_1+r\tPhi_1}{1+r}=\tr_{\cK_2}{\Psi}\right. \\
& \left. \frac{\Phi_2+r\tPhi_2}{1+r}=\tr_{\cK_1}{\Psi}\right. \\
& \left. \Psi \in\Ch(\cH,\cK_1 \otimes \cK_2)\right. \\
& \left. \tPhi_i\in\Ch(\cH,\cK_i) \quad i = 1,2. \right.
\end{split}
\end{equation}
Note that via the Choi representation, the incompatibility robustness of quantum channels can be cast as an SDP as well, and that this definition also naturally generalizes to larger sets of channels.

\section{Charaterizing the compatibility of quantum instruments}\label{sec:main-charac-comp-ins}

In this section, we initiate the quantitative characterization of parallel incompatibility of instruments, along similar lines to the characterization of measurement and channel incompatibility through incompatibility robustness. First, we define the generalized incompatibility robustness of quantum instruments in an analogous way to measurements and channels.
\begin{definition}\label{Def:ROI-instruments}
The generalized incompatibility robustness of two quantum instruments $\cI_1=\{\Phi^1_{x}:\cS(\cH)\rightarrow\cL^+(\cK_1)\}_{x \in \Omega_1}$ and $\cI_2=\{\Phi^2_{y}:\cS(\cH)\rightarrow\cL^+(\cK_2)\}{y \in \Omega_2}$ is given by
\begin{equation}\label{Eq:ROI-instruments}
\begin{split}
R_{I}(\cI_1,\cI_2 )=\min~  & \left.  r\right. \\
\text{s.t.} ~ & \left. \frac{\Phi^1_{x}+r\tPhi^{1}_{x}}{1+r}= \sum_y \tr_{\cK_2} \Psi_{xy} \right. \\
& \left. \frac{\Phi^2_{y}+r\tPhi^{2}_{y}}{1+r}= \sum_x \tr_{\cK_1} \Psi_{xy} \right. \\
& \left. \cI = \{ \Psi_{xy} \} \in\Ins(\Omega_1 \times \Omega_2,\cH,\cK_1 \otimes \cK_2)\right. \\
& \left. \tI_1 = \{ \tPhi^1_x \} \in \Ins( \Omega_1, \cH, \cK_1 ) \right. \\
& \left. \tI_2 = \{ \tPhi^2_y \} \in \Ins( \Omega_2, \cH, \cK_2 ) \right.
\end{split}
\end{equation}
\end{definition}
Similarly to the generalized incompatibility robustness of measurements and channels, the generalized incompatibility robustness of instruments can be cast as an SDP (through the Choi states of the individual CP maps), and the definition naturally generalizes to more than two instruments.

In the remainder of this section, we prove some basic characteristics of the incompatibility robustness of quantum instruments. Note that in the following, by ``incompatibility robustness'' we mean generalized incompatibility robustness, unless otherwise specified.

\subsection{Bounds on the incompatibility robustness}

First, we show that the incompatibility robustness of two quantum instruments is lower bounded by the incompatibility robustness of their induced POVMs and their induced channels.

\begin{theorem}
Consider two instruments $\cI_1=\{\Phi^1_x\}\in\Ins(\Omega_{\cI_1},\cH,\cK_1)$ and $\cI_2=\{\Phi^2_y\}\in\Ins(\Omega_{\cI_2},\cH,\cK_2)$. Let us denote the POVM induced by $\cI_i$ by $A_i$. Furthermore, consider the induced channels, $\Phi^1=\sum_{x\in\Omega_{\cI_1}}\Phi^2_x$, and $\Phi^2=\sum_{y\in\Omega_{\cI_2}}\Phi^2_y$.
 Then we have that $R_I(\cI_1,\cI_2)\geq \max \{R_M(A_1,A_2),R_{C}(\Phi_1,\Phi_2)\}$. \label{Th:ROI-lower-bou}
\end{theorem}

\begin{proof}
Suppose that $R_I(\cI_1,\cI_2)=r$. Then, following from the definition of incompatibility robustness of instruments in Definition~\ref{Def:ROI-instruments}, there exist quantum instruments $\cI = \{ \Psi_{xy} \} \in \Ins( \Omega_{\cI_1} \times \Omega_{\cI_2}, \cH, \cK_1 \otimes \cK_2)$, $\tI_1 = \{ \tPhi^1_x \} \in \Ins( \Omega_{\cI_1}, \cH, \cK_1)$ and $\tI_2 = \{ \tPhi^2_y \} \in \Ins(\Omega_{\cI_2}, \cH, \cK_2)$ such that
\begin{align}\label{eq:ROI_inst_r1}
\frac{ \Phi^1_x + r \tPhi^1_x }{1+r} = \sum_y \tr_{\cK_2} \Psi_{xy} \quad \forall x \\ \label{eq:ROI_inst_r2}
\frac{ \Phi^2_y + r \tPhi^2_y }{1+r} = \sum_x \tr_{\cK_1} \Psi_{xy} \quad \forall y.
\end{align}
Summing up the first equation over $x$ and the second one over $y$, we see that there exist quantum channels $\tPhi^{1}$ and $\tPhi^{2}$ such that $(\Phi^1+r\tPhi^{1})/(1+r)$ and $(\Phi^2+t\tPhi^{2})/(1+r)$ are compatible. Therefore, by the definition of incompatibility robustness of channels in Eq.~\eqref{Def:ROI-channels}, we have that $r \geq R_{C}(\Phi_1,\Phi_2)$.

In a similar way, one can show that $r \geq R_M(A_1,A_2)$. Let us take the dual of the operations in Eqs.~\eqref{eq:ROI_inst_r1} and \eqref{eq:ROI_inst_r2}, and notice that $(\tr_{\cK_2} \Psi_{xy})^\ast (X) = \Psi^\ast_{xy} (X \otimes \Id_{\cK_2})$ for all $X \in \cL(\cK_1)$ and similarly $(\tr_{\cK_1} \Psi_{xy})^\ast (X) = \Psi^\ast_{xy} (\Id_{\cK_1} \otimes X)$ for all $X \in \cL(\cK_2)$. We then apply the dual of Eq.~\eqref{eq:ROI_inst_r1} to $\Id_{\cK_1}$ and the dual of Eq.~\eqref{eq:ROI_inst_r2} to $\Id_{\cK_2}$, which yields
\begin{align}
\frac{A_1(x) + r \tA_1(x)}{1+r} = \sum_y G(x,y) \\
\frac{A_2(y) + r \tA_2(y)}{1+r} = \sum_x G(x,y), 
\end{align}
where we have defined the POVMs $\tA_1(x) = (\tPhi^1_x)^\ast(\Id_{\cK_1})$, $\tA_2(y) = (\tPhi^2)_y)^\ast(\Id_{\cK_2})$, and $G(x,y) = \Psi^\ast_{xy}(\Id_{\cK_1} \otimes \Id_{\cK_2})$. From the definition of the incompatibility robustness of measurements in Eq.~\eqref{Def:ROI-measurements} it is clear then that $r \ge R_M(A_1,A_2)$.

In summary, we have that $R_I(\{\cI_1,\cI_2\})\geq \max\{R_M(A_1,A_2),R_{C}(\Phi_1,\Phi_2)\}$, which is exactly the statement of the theorem.
\end{proof}

The following theorem establishes a universal upper bound on the incompatibility robustness of two instruments.

\begin{theorem}
Consider two arbitrary instruments $\cI_1=\{\Phi^1_x\}\in\Ins(\Omega_{\cI_1},\cH,\cK_1)$ and $\cI_2=\{\Phi^2_y\}\in\Ins(\Omega_{\cI_2},\cH,\cK_2)$. Then $ R_I(\cI_1,\cI_2) \le 1$.
\end{theorem}

\begin{proof}
Consider two quantum states $\eta_1=\ket{\eta_1}\bra{\eta_1}\in\cS(\cK_1)$ and $\eta_2=\ket{\eta_2}\bra{\eta_2}\in\cS(\cK_2)$.
Let us define a quantum instrument $\cI=\{\Psi_{xy}\}$, where $x\in\Omega_{\cI_1}$ and $y\in\Omega_{\cI_2}$, given by its action
\begin{align}
\Psi_{xy}(\rho)=\frac{1}{2} \left(\Phi^1_x(\rho)\otimes\frac{1}{|\Omega_{\cI_2}|}\eta_2+\frac{1}{|\Omega_{\cI_1}|}\eta_1\otimes\Phi^2_y(\rho) \right),
\end{align}
which clearly defines a quantum instrument.

Now, consider the quantum instruments $\tI^{1}=\{\tPhi^{1}_x=\sum_{y\in\Omega_{\cI_2}}\tr_{\cK_2}\Psi_{xy}=\frac{\Phi^1_x+\Phi^{\eta_1}_x}{2}\}_{x\in\Omega_{\cI_1}}$ and $\tI^{2}=\{\tPhi^{2}_y=\sum_{x\in\Omega_{\cI_1}}\tr_{\cK_1}\Psi_{xy}=\frac{\Phi^2_y+\Phi^{\eta_2}_y}{2}\}_{y\in\Omega_{\cI_2}}$ where $\mathfrak{T}_{\eta_1}=\{\Phi^{\eta_1}_x\}_{x\in\Omega_{\cI_1}}$ and $\mathfrak{T}_{\eta_2}=\{\Phi^{\eta_2}_y\}_{y\in\Omega_{\cI_2}}$ are  the trash-and-prepare instruments defined via $\Phi^{\eta_1}_x(\rho)=\frac{\tr(\rho)}{|\Omega_{\cI_1}|}\eta_1$ and $\Phi^{\eta_2}_y(\rho)=\frac{\tr(\rho)}{|\Omega_{\cI_2}|}\eta_2$ for all $\rho\in\cS(\cH)$, $x\in\Omega_{\cI_1}$ and $y\in\Omega_{\cI_2}$. By definition, $\tI_1$ and $\tI_2$ are compatible via the joint instrument $\cI$. Furthermore, they are noisy versions of the instruments $\cI_1$ and $\cI_2$ with noise parameter $r = 1$ and noise instruments $\mathfrak{T}_{\eta_1}$ and $\mathfrak{T}_{\eta_2}$. Therefore, from the definition of incompatibility robustness in Definition~\ref{Def:ROI-instruments}, we get that $R_I(\cI_1,\cI_2) \le 1$.

\end{proof}

\subsection{Post-processing as a free operation}

In the next theorem we prove that post-processing preserves the compatibility of quantum instruments. That is, one may view post-processing as a free operation in a potential resource theory of incompatibility.

\begin{theorem}
Consider two arbitrary instruments $\cI_1=\{\Phi^1_x\}\in\Ins(\Omega_{\cI_1},\cH,\cK_1)$ and $\cI_2=\{\Phi^2_y\}\in\Ins(\Omega_{\cI_2},\cH,\cK_2)$ and  another two instruments $\tI_1=\{\tPhi^{1}_{\tx}\}\in\Ins(\Omega_{\tI_1},\cH,\tK_1)$ and $\tI_2=\{\tPhi^{2}_{\ty}\}\in\Ins(\Omega_{\tI_2},\cH,\tK_2)$ such that $\tI_1\lesssim\cI_1$ and $\tI_2\lesssim\cI_2$. Then, if $\cI_1$ and $\cI_2$ are compatible then $\tI_1$ and $\tI_2$ are also compatible. \label{Th:post-proc-in-comp-pres}
\end{theorem}

\begin{proof}
Since $\cI_1$ and $\cI_2$ are compatible, there exists an instrument $\cI=\{\Phi_{xy}\}$ such that Eqs.~\eqref{Eq:def-par-comp-1} and \eqref{Eq:def-par-comp-2} hold. Also, since $\tI_1\lesssim\cI_1$ and $\tI_2\lesssim\cI_2$, there exist two sets of instruments $\{\cR^{1,x}=\{R^{1,x}_{\tx}\}\in\Ins({\Omega}_{\tI_1},\cK_1,\tK_1)\}_{x\in\Omega_{\cI_1}}$ and $\{\cR^{2,y}=\{R^{2,y}_{\ty}\}\in\Ins({\Omega}_{\tI_2},\cK_2,\tK_2)\}_{y\in\Omega_{\cI_2}}$ such that

\begin{align}
\tPhi^{1}_{\tx}=\sum_{x\in\Omega_{\cI_1}} R^{1,x}_{\tx}\circ\Phi^1_{x} \quad  \forall\tx\in\Omega_{\tI_1} \label{Eq:post-proc-I-1-p},\\
\tPhi^{2}_{\ty}=\sum_{y\in\Omega_{\cI_2}} R^{2,y}_{\ty}\circ\Phi^2_{y} \quad \forall \ty\in\Omega_{\tI_2} \label{Eq:post-proc-I-2-p}.
\end{align} 
Let us define the instrument $\tI=\{\tPhi_{\tx \ty}\}_{\tx\in\Omega_{\tI_1},\ty\in\Omega_{\tI_2}}$ given by its action
 \begin{align}
\tPhi_{\tx \ty} = \sum_{x \in \Omega_{\cI_1}} \sum_{y \in \Omega_{\cI_2}} (R^{1,x}_{\tx} \otimes R^{2,y}_{\ty}) \circ \Phi_{xy}.
 \end{align}
Since $\sum_{\tx \in \Omega_{\tI_1}} R^{1,x}_{\tx}=:R^{1,x}$ and $\sum_{\ty \in \Omega_{\tI_2}} R^{2,y}_{\ty} =: R^{2,y}$ are both quantum channels and therefore, trace preserving, we have that
\begin{align}
\sum_{\ty \in \Omega_{\tI_2}} \tr_{\tK_2} (R^{1,x}_{\tx} \otimes R^{2,y}_{\ty}) (.) = R^{1,x}_{\tx} \circ \tr_{\cK_2} (.) \\
\sum_{\tx \in \Omega_{\tI_1}} \tr_{\tK_1} (R^{1,x}_{\tx} \otimes R^{2,y}_{\ty}) (.) = R^{2,y}_{\ty} \circ \tr_{\cK_1} (.).
\end{align}
Using these relations, it is easy to verify that
\begin{align}
\sum_{\ty \in \Omega_{\tI_2}} \tr_{\tK_2} \tPhi_{\tx \ty}=\tPhi^{1}_{\tx} \\
\sum_{\tx \in \Omega_{\tI_1}} \tr_{\tK_1} \tPhi_{\tx \ty}=\tPhi^{2}_{\ty}.
\end{align}
Therefore, $\tI_1$ and $\tI_2$ are compatible.
\end{proof}

As a consequence of the above theorem, given a pair of compatible instruments, every pair of instruments post-processing equivalent to this pair is also compatible:

\begin{corollary}\label{cor:compatible_equiv}
Consider two arbitrary instruments $\cI_1$ and $\cI_2$, and another two instruments $\tI_1$ and $\tI_2$ such that $\tI_1 \sim \cI_1$ and $\tI_2 \sim \cI_2$. Then, $\cI_1$ and $\cI_2$ are compatible if and only if $\tI_1$ and $\tI_2$ are compatible.
\end{corollary}

In the following, we will show that trash-and-prepare instruments are compatible with any other instrument with the same input space. As a preparation, we prove the following lemma, noting that a quantum channel can be considered as one-outcome instrument.
\begin{lemma}\label{prop:identity_trash-and-prepare}
The identity channel $\IdC_{\cH}\in\mathbbm{Ch}(\cH,\cH)$, viewed as a one-outcome instrument $\IdC_{\cH}\in\Ins(\{1\},\cH,\cH)$ is compatible with any trash-and-prepare instrument $\mathfrak{J^{\prime}}_{T}\in\mathbbm{In}(\Omega_{\mathfrak{J^{\prime}}_{T}},\cH,\cK)$.
\end{lemma}
\begin{proof}
Let us denote the fixed outcome of the identity channel $\IdC_{\cH}$ by `$1$'. Consider the trivial measurement $T=\{T(x)=p_x\Id\}$, where $\{p_x\}$ is a probability distribution. Then, $\mathfrak{J^{\prime}}_{T}$ has the form $\mathfrak{J^{\prime}}_{T}=\{\Phi_x\}$ where for any quantum state $\rho\in\cS(\cH)$, $\Phi_x(\rho)=p_x\eta_x$ where $\eta_x\in\cS(\cK)$ for all $x\in\Omega_{\mathfrak{J^{\prime}}_{T}}$. Now consider the quantum instrument $\cI=\{\Phi^{\prime}_{x1}\}$ where where for any quantum state $\rho\in\cS(\cH)$
\begin{align}
\Phi^{\prime}_{x1}(\rho)=p_x\eta_x\otimes\rho
\end{align}
for all $x\in\Omega_{\mathfrak{J^{\prime}}_{T}}$. Then, from Eqs.~\eqref{Eq:def-par-comp-1} and \eqref{Eq:def-par-comp-2}, we get that $\cI$ is a joint instrument for $\IdC_{\cH}$ and $\mathfrak{J^{\prime}}_{T}$ and hence, they are compatible.
\end{proof}

From the fact that any instrument can be post-processed from the identity channel \cite{LS21}, Theorem \ref{Th:post-proc-in-comp-pres} and Lemma \ref{prop:identity_trash-and-prepare}, the following proposition follows directly.

\begin{proposition}\label{cor:compatible_trash-prep}
Any instrument $\cI\in\Ins(\Omega_{\cI},\cH,\cK)$ is compatible with any trash-and-prepare instrument $\mathfrak{J^{\prime}}_{T}\in\Ins(\Omega_{\mathfrak{J^{\prime}}_{T}},\cH,\cK^{\prime})$.
\end{proposition}

The following theorem establishes an even stronger relation between incompatibility robustness and the post-processing preorder: the incompatibility robustness of quantum instruments is monotonic under post-processing.

\begin{theorem}
Consider two arbitrary instruments $\cI_1 = \{ \Phi^1_x \} \in \Ins (\Omega_{\cI_1}, \cH, \cK_1)$ and $\cI_2=\{\Phi^2_y\}\in\Ins( \Omega_{\cI_2}, \cH, \cK_2)$, and another two instruments $\tI_1 = \{ \tPhi^{1}_{\tx} \} \in \Ins ( \Omega_{\tI_1}, \cH, \tK_1)$ and $\tI_2= \{ \tPhi^{2}_{\ty} \} \in \Ins( \Omega_{\tI_2}, \cH, \tK_2)$ such that $\tI_1 \lesssim \cI_1$ and $\tI_2 \lesssim \cI_2$. Then, $R_I( \tI_1,\tI_2 ) \le R_I( \cI_1,\cI_2 )$. \label{Th:post-proc-in-comp-ROI}
\end{theorem}

\begin{proof}

Since $\tI_1 \lesssim \cI_1$ and $\tI_2 \lesssim \cI_2$, there exist two sets of instruments $\{\cR^{1,x}=\{R^{1,x}_{\tx} \} \in\Ins ( {\Omega}_{\tI_1}, \cK_1, \tK_1) \}_{x\in\Omega_{\cI_1}}$ and $\{\cR^{2,y}=\{R^{2,y}_{\ty} \} \in \Ins ({\Omega}_{\tI_2}, \cK_2, \tK_2) \}_{y\in\Omega_{\cI_2}}$ such that Eqs.~\eqref{Eq:post-proc-I-1-p} and \eqref{Eq:post-proc-I-2-p} hold.

Suppose that $R_I( \cI_1,\cI_2 ) = r$. Then there exist two instruments $\{J^1_x\}\in\Ins(\Omega_{\cI_1},\cH,\cK_1)$ and $\{J^2_y\}\in\Ins(\Omega_{\cI_2},\cH,\cK_2)$ such that the instruments $\cJ_1:=\{\frac{\Phi^1_x+ r J^1_x}{1+r}\}\in\Ins(\Omega_{\cI_1},\cH,\cK_1)$ and $\cJ_2:=\{\frac{\Phi^2_x+ r J^2_x}{1+r}\}\in\Ins(\Omega_{\cI_2},\cH,\cK_2)$ are compatible. Consider the post-processings of $\cJ_1$ and $\cJ_2$ given by $\tilde{\cJ}_1=\{\frac{\tPhi^{1}_{\tx}+ r \tilde{J}^{1}_{\tx}}{1+r}\}\in\Ins(\Omega_{\tI_1},\cH,\tK_1)$ and $\tilde{\cJ}_2=\{\frac{\tPhi^{2}_{\ty}+ r \tilde{J}^{2}_{\ty}}{1+r}\}\in\Ins(\Omega_{\tI_2}, \cH, \tK_2)$, where
\begin{align}
\tilde{J}^{1}_{\tx} = \sum_{x\in\Omega_{\cI_1}} R^{1,x}_{\tx} \circ J^1_{x},\\
\tilde{J}^{2}_{\ty} = \sum_{y\in\Omega_{\cI_2}} R^{2,y}_{\ty} \circ J^2_{y}.
\end{align} 
Then from Theorem \ref{Th:post-proc-in-comp-pres}, we get that $\tilde{\cJ}_1$ and $\tilde{\cJ}_2$ are compatible, since $\tilde{\cJ}_1 \lesssim \cJ_1$ and $\tilde{\cJ}_2 \lesssim \cJ_2$. Thus, from the definition of the incompatibility robustness of instruments, we get that $R_I( \tI_1,\tI_2 ) \le r$.

\end{proof}

As a consequence of the above theorem, every post-processing equivalent pair of instruments has the same incompatibility robustness:

\begin{corollary}
Consider two arbitrary instruments $\cI_1$ and $\cI_2$, and another two instruments $\tI_1$ and $\tI_2$ such that $\tI_1 \sim \cI_1$ and $\tI_2 \sim \cI_2$. Then, $R_I(\cI_1, \cI_2 )= R_I( \tI_1, \tI_2 )$. \label{Coro:post-proc-in-comp-ROI}
\end{corollary}

One might also consider other operations that preserve compatibility in a potential resource theory of incompatibility. Here we show that the operation that was shown to be a free operation for traditional compatibility \cite{JC21} is not a free operation for parallel compatibility.

Ref.~\cite{JC21} uses the terminology \textit{programmable instrument devices} (PIDs) for a set of instruments that induce the same quantum channel. Furthermore, a PID is called \textit{simple} if the instruments that comprise the PID are traditionally compatible. Hence, since traditionally compatible instruments induce the same quantum channel \cite{MF22}, simple PIDs correspond exactly to traditionally compatible instruments. The authors of Ref.~\cite{JC21} introduce a family of maps that map PIDs to PIDs, moreover, map simple PIDs to simple PIDs, and they call this family \textit{free PID supermaps}. For the sake of clarity, we recall this definition here:

\begin{definition}
A free PID supermap maps a set of instruments $\big\{ \cI_i = \{ \Phi^i_x \}_{x \in \Omega_{\cI_i}} \in \Ins( \Omega_{\cI_i}, \cH, \cK ) \big\}_{i = 1}^{n_\cI}$ to another set of instruments $\big\{ \cJ_j = \{ \Psi^j_y \}_{y \in \Omega_{\cJ_j}} \in \Ins( \Omega_{\cJ_j}, \tH, \tK ) \big\}_{j = 1}^{n_{\cJ}}$. The map is given by
\begin{equation}
\Psi^j_y = \sum_{s,i,x} q_{y|x,s} \Gamma^j_{is} \circ ( \Phi^i_x \otimes \mathrm{id} ) \circ \Lambda,
\end{equation}
where $\Lambda: \cS(\tH) \to \cS(\cH \otimes \cQ)$ is a quantum channel, $\mathrm{id}$ is the identity channel on $\cQ$, $s \in \{1, \ldots, n_S\} =: [n_S]$, $\big\{ \Gamma_j := \{ \Gamma^j_{is} \}_{i,s} \in \Ins( [n_S] \otimes [n_\cI], \cK \otimes \cQ, \tK) \big\}_{j = 1}^{n_\cJ}$ is an arbitrary simple PID (set of traditionally compatible instruments), and $q_{y|x,s}$ is a probability distribution for every $x$ and $s$.
\end{definition}
The following two statements were proven in Ref.~\cite[Theorem 1]{JC21}.
\begin{theorem}\label{thm:PID}
Free PID supermaps map every set of traditionally compatible instruments to a set of traditionally compatible instruments. Moreover, every set of traditionally compatible instruments can be mapped to any other set of traditionally compatible instruments using a free PID supermap.
\end{theorem}
The second statement is relatively easy to see by noticing that the definition of a free PID supermap includes applying a set of traditionally compatible instruments. Hence, one can just throw away the output of the original set and apply the new set.

The above theorem essentially states that free PID supermaps are free operations for traditional compatibility. Below, we show that this is not the case for parallel compatibility.
\begin{theorem}
Free PID supermaps are not free operations for parallel compatibility. In particular, free PID supermaps might map parallel compatible instruments to parallel incompatible instruments.
\end{theorem}

\begin{proof}
Consider the following two instruments
\begin{gather}
\cI = \left\{ \Phi_x : \cS(\cH) \to \cL^+(\cH) ~|~ \Phi_x(\rho) = \frac{ \tr (\rho) }{n} \eta \right\}_{x = 1}^n \\
\cJ = \left\{ \Psi_x: \cS(\cH) \to \cL^+(\cH) ~|~ \Psi_x(\rho) = \frac1n \rho \right\}_{x=1}^n,
\end{gather}
where $\eta$ is some fixed state on $\cH$. It is straightforward to verify that two copies of $\cI$ are traditionally compatible via the joint instrument
\begin{equation}
\cI_t = \left\{ \Phi^t_{xy} : \cS(\cH) \to \cL^+(\cH) ~|~ \Phi^t_{xy}(\rho) = \frac{ \tr (\rho) }{n^2} \eta \right\}_{x,y = 1}^n
\end{equation}
Similarly, two copies of $\cJ$ are also traditionally compatible via the joint instrument
\begin{equation}
\cJ_t = \left\{ \Psi^t_{xy} : \cS(\cH) \to \cL^+(\cH) ~|~ \Psi^t_{xy} (\rho) = \frac{1}{n^2} \rho \right\}_{x,y=1}^n
\end{equation}
Therefore, following from Theorem~\ref{thm:PID}, there exists a free PID supermap that maps two copies of $\cI$ to two copies of $\cJ$.

One can also show that two copies of $\cI$ are parallel compatible via the joint instrument
\begin{equation}
\cI_p = \left\{ \Phi^p_{xy} : \cS(\cH) \to \cL^+(\cH \otimes \cH) ~|~ \Phi^p_{xy}(\rho) = \frac{ \tr (\rho) }{n^2} \eta \otimes \eta \right\}_{x,y = 1}^n
\end{equation}

On the other hand, two copies of $\cJ$ are not parallel compatible, which is a direct consequence of the no-cloning theorem, or alternatively, that the identity channel is not compatible with itself \cite{MF22}. Therefore, the free PID supermap mapping two copies of $\cI$ to two copies of $\cJ$ maps a pair of parallel compatible instruments to a pair of parallel incompatible instruments. As such, free PID supermaps are in general not free operations for parallel compatibility.
\end{proof}

\begin{remark}
In general, the above theorem is a consequence of Theorem \ref{thm:PID}, the fact that there exist instruments that are both traditionally compatible and parallel compatible, and the fact that there exist instruments that are traditionally compatible but not parallel compatible~\cite{MF22}.
\end{remark}

Concluding this section, we have identified a natural class of free operations for parallel compatibility, namely the class of post-processing of quantum instruments from Ref.~\cite{LS21}. At the same time, we showed that another class of operations---the free PID supermaps of Ref.~\cite{JC21}---that are natural free operations for traditional compatibility are not free operations for parallel compatibility. It could be an interesting research direction to try and identify more, physically motivated classes of free operations for parallel compatibility and develop a resource theory of instrument incompatibility based on these free operations.

\subsection{Cases for which the bounds are tight}

The following theorem establishes that the lower bound on the incompatibility robustness of instruments in terms of the incompatibility robustness of the induced POVMs and channels from Theorem \ref{Th:ROI-lower-bou} is tight for special measure-and-prepare instruments.

\begin{theorem}
 The lower bound in Theorem \ref{Th:ROI-lower-bou} is tight for any pair of special measure-and-prepare instruments $\mathfrak{I}_{A_1}$ and $\mathfrak{I}_{A_2}$, where $A_1 \in \M(\Omega_{A_1}, \cH)$ and $A_2\in\M( \Omega_{A_2}, \cH)$. \label{Th:ROI-tight-lower-bou}
\end{theorem}

\begin{proof}
Suppose that $R_M( A_1,A_2 ) = r$. Then there exist $A^{\prime}_1\in \M(\Omega_{A_1}, \cH)$ and $A^{\prime}_2\in \M( \Omega_{A_2}, \cH)$ such that $\bar{A}_1=\{\bar{A}_1(x)=\frac{A_1(x)+ r A_1^{\prime}(x)}{1+r}\}_{x\in \Omega_{A_1}}$ and $\bar{A}_2=\{\bar{A}_2(y)=\frac{A_2(y)+ r A_2^{\prime}(y)}{1+r}\}_{y\in \Omega_{A_2}}$ are compatible. That is, there exists a joint measurement $G\in \M( \Omega_{A_1} \times \Omega_{A_2}, \cH)$ such that
\begin{align}
\bar{A}_1(x)=\sum_{y}G(x,y) \quad \forall x\in\Omega_{A_1} \\
\bar{A}_2(y)=\sum_{x}G(x,y) \quad \forall y\in\Omega_{A_2}.
\end{align}

Now consider the special measure-and-prepare quantum instrument $\cI=\{\Phi_{xy}:\cS(\cH)\rightarrow \cL^+(\cK_1\otimes\cK_2)\}_{x\in\Omega_{A_1},y\in\Omega_{A_2}}$  given by
\begin{align}
\Phi_{xy}(\rho)=\tr[\rho G(x,y)]\ketbraq{x} \otimes \ketbraq{y}.
\end{align}
Then, for all $x\in\Omega_{A_1}$ we have that
\begin{align}
\sum_y\tr_{\cK_2}\Phi_{xy}(\rho)=&\sum_y\tr[\rho G(x,y)]\ketbraq{x}\nonumber\\
=&\tr[\rho \bar{A}_1(x)]\ketbraq{x}\nonumber\\
=&\tr \left[\rho \frac{A_1(x)+ r A_1^{\prime}(x)}{1+r} \right]\ketbraq{x}\nonumber\\
=&\frac{\tr[\rho A_1(x)]\ketbraq{x}+ r \tr[\rho A_1^{\prime}(x)]\ketbraq{x}}{1+r},
\end{align}
and for all $y\in\Omega_{A_2}$ we have that
\begin{align}
\sum_x\tr_{\cK_1}\Phi_{xy}(\rho)=&\sum_x\tr[\rho G(x,y)]\ketbraq{y}\nonumber\\
=&\tr[\rho \bar{A}_2(y)]\ketbraq{y}\nonumber\\
=&\tr \left[\rho \frac{A_2(y)+ r A_2^{\prime}(y)}{1+r} \right]\ketbraq{y}\nonumber\\
=&\frac{\tr[\rho A_2(y)]\ketbraq{y}+ r \tr[\rho A_2^{\prime}(y)]\ketbraq{y}}{1+r}.
\end{align}
Hence, from Definition \ref{Def:spec-meas-prep} and Definition \ref{Def.parallel_comp}, we get that $r=R_M( A_1,A_2 ) \geq R_I( \mathfrak{I}_{A_1}, \mathfrak{I}_{A_2} )$. Since from Theorem~\ref{Th:ROI-lower-bou}, we have $R_M( A_1,A_2 ) \leq R_I( \mathfrak{I}_{A_1}, \mathfrak{I}_{A_2} )$, we establish that $R_M( A_1,A_2 )= R_I( \mathfrak{I}_{A_1}, \mathfrak{I}_{A_2} )$. Also, from Theorem \ref{Th:ROI-lower-bou}, we know that $R_I( \mathfrak{I}_{A_1}, \mathfrak{I}_{A_2} )\geq R_C( \Gamma_{A_1}, \Gamma_{A_2} )$. Therefore, we conclude that $R_M(A_1,A_2 )\geq R_C( \Gamma_{A_1}, \Gamma_{A_2} )$, and in summary, the lower bound in Theorem \ref{Th:ROI-lower-bou} is tight.
\end{proof}

The following lemma from Ref.~\cite{LS21} allows us to generalize the previous result to arbitrary measure-and-prepare instruments.

\begin{lemma}
For any measurement $A\in \M( \Omega_A, \cH)$, a generic measure-and-prepare instrument $\mathfrak{J}^{\prime}_{A}$ is a post-processing of the special measure-and-prepare instrument $\mathfrak{J}_{A}$, that is, $\mathfrak{J}^{\prime}_{A}\lesssim\mathfrak{J}_{A}$ \cite[Proposition 10]{LS21}.\label{Le:meas-prep-spec-meas-prep}
\end{lemma}

\begin{corollary}
 The lower bound in Theorem \ref{Th:ROI-lower-bou} is tight for any pair of generic measure-and-prepare instruments $\mathfrak{I}^{\prime}_{A_1}$ and $\mathfrak{I}^{\prime}_{A_2}$ for any $A_1\in \M( \Omega_{A_1}, \cH)$ and $A_2\in \M( \Omega_{A_2}, \cH)$. \label{Coro:ROI-tight-lower-bou}
\end{corollary}

\begin{proof}
From Lemma~\ref{Le:meas-prep-spec-meas-prep}, Theorem \ref{Th:post-proc-in-comp-ROI} and Theorem~\ref{Th:ROI-tight-lower-bou} we have $R_I( \mathfrak{I}^{\prime}_{A_1},\mathfrak{I}^{\prime}_{A_2} )\leq R_I( \mathfrak{I}_{A_1},\mathfrak{I}_{A_2} )=R_M( A_1,A_2 )$, and from Theorem~\ref{Th:ROI-lower-bou}, we have $R_I( \mathfrak{I}^{\prime}_{A_1},\mathfrak{I}^{\prime}_{A_2} )\geq R_M( A_1,A_2 )$. Therefore, we conclude that $R_I( \mathfrak{I}^{\prime}_{A_1},\mathfrak{I}^{\prime}_{A_2} )=R_M( A_1,A_2 )$.
\end{proof}

\subsection{Indecomposable instruments}

Indecomposable instruments have been introduced in Ref.~\cite{LS21}. They play a fundamental role in the post-processing preorder, since (while these instruments are not in general maximal in the preorder) every instrument can be obtained from an indecomposable instrument via post-processing. This property is analogous to rank-1 POVMs, in the sense that every POVM can be obtained from a rank-1 POVM via classical post-processing. For the sake of clarity, we recall the definition and a basic characterisation of indecomposable instruments from Ref.~\cite{LS21}.
\begin{definition}
A quantum operation $\Phi$ is called indecomposable if $\Phi = \Psi + \Psi'$ for some quantum operations $\Psi$ and $\Psi'$ implies $\Psi = \mu \Phi$ and $\Psi' = \mu' \Phi$ for some $\mu, \mu' >0$. A quantum instrument is called indecomposable if all of its nonzero operations are indecomposable.
\end{definition}
\begin{proposition}
A quantum operation is indecomposable if and only if it has Kraus rank one.
\end{proposition}

To show that every instrument is a post-processing of an indecomposable instrument, one needs to introduce the concept of a \textit{detailed instrument} \cite{LS21}.
\begin{definition}
Consider a quantum instrument $\cI = \{ \Phi_x \}$ with Kraus decomposition $\Phi_x(\rho) = \sum_{i=1}^{n_x} K_{ix} \rho K^\dagger_{ix}$. The detailed instrument of $\cI$ is defined as $\hat{\cI} = \{ \hat{\Phi}_{ix}(\rho) = K_{ix} \rho K^\dagger_{ix}\}$.
\end{definition}
It is straightforward to see that for every instrument $\cI$ we have that $\cI \lesssim \hat{\cI}$ \cite{LS21}, and clearly every detailed instrument is indecomposable. In Ref.~\cite{LS21}, the authors provide a class of instruments for which the converse is true as well, and hence the instrument is post-processing equivalent to its detailed instrument.
\begin{proposition}\label{prop:detailed_equiv}
An instrument $\cI = \{ \Phi_x \} \in \Ins(\Omega,\cH, \cK)$ with Kraus decomposition $ \Phi_x(\rho) = \sum_{i=1}^{n_x} K_{ix} \rho K^\dagger_{ix}$ is post-processing equivalent to its detailed instrument if $K^\dagger_{ix} K_{jx} = 0$ for all $i \neq j$ and $x \in \Omega$.
\end{proposition}

Since indecomposable instruments are fundamental in the post-processing preorder, it is useful to characterise compatible sets of indecomposable instruments. Here we provide a family of compatible pairs of indecomposable instruments constructed from compatible measurements.
\begin{theorem}\label{thm:indecomp_comp}
For every pair of compatible measurements $A \in \M( \Omega_A, \cH)$ and $B \in \M( \Omega_B, \cH)$ there exists a pair of compatible indecomposable instruments.
\end{theorem}

\begin{proof}
We base our construction on the pair of compatible instruments for a given pair $(A,B)$ of compatible measurements from Ref.~\cite{MF22}. While these instruments are in general not indecomposable, we show using Proposition \ref{prop:detailed_equiv} that they are post-processing equivalent to their detailed instruments. Therefore, following from Corollary \ref{cor:compatible_equiv}, the detailed instruments (which are indecomposable) are also compatible.

Let us recall the construction from Ref.~\cite{MF22}. Consider two POVMs, $\{A(x)\}_{x \in \Omega_A}$ and $\{B(y)\}_{y \in \Omega_B}$ on $\cH$ such that they are compatible, that is, there exists a POVM $\{ G(x,y) \}_{x \in \Omega_A, y \in \Omega_B}$ on $\cH$ such that
\begin{gather}
A(x) = \sum_y G(x,y) \quad \forall x \\
B(y) = \sum_x G(x,y) \quad \forall y.
\end{gather}
Now consider a Naimark dilation of the joint measurement $G$, that is, a projective measurement $\Pi(x,y)$ on a Hilbert space $\cH \otimes \cK$ such that
\begin{equation}
\tr[ (\rho \otimes \ketbraq{0}) \Pi(x,y) ] = \tr[ \rho G(x,y) ] \quad \forall \rho \in \cS(\cH)
\end{equation}
for some $\ket{0} \in \cK$. Furthermore, consider a rank-1 fine-graining of $\Pi(x,y)$, given by
\begin{equation}
\Pi(x,y) = \sum_{z_{xy} \in P(x,y)} \ketbraq{ \phi_{z_{xy}} },
\end{equation}
where $P(x,y)$ is a set indexing the rank-1 projections that comprise $\Pi(x,y)$. It is clear that $\braket{ \phi_{z_{xy}} | \phi_{z'_{x'y'}} } = \delta_{x,x'} \delta_{y,y'} \delta_{ z_{xy}, z'_{x'y'} }$.

Take the $A$-compatible instrument
\begin{equation}
\begin{split}
\cI_A = & \left. \Big\{ \Phi^A_x: \cS(\cH) \to \cL^+_{\cH \otimes \cK} ~| \right. \\  \Phi^A_x(\rho) = & \left. \sum_y \sum_{z_{xy} \in P(x,y)} \ketbraq{ \phi_{z_{xy}} }( \rho \otimes \ketbraq{0} ) \ketbraq{ \phi_{z_{xy}} } \right. \\
= & \left. \sum_y \sum_{z_{xy} \in P(x,y)} K^A_{(y,z_{xy}),x} \rho (K^A_{(y,z_{xy}),x})^\dagger \Big\} \right.
\end{split}
\end{equation}
where $K^A_{(y,z_{xy}),x} = \ketbraq{ \phi_{z_{xy}} } (\Id_\cH \otimes \ket{0})$, and the $B$-compatible instrument
\begin{equation}
\begin{split}
\cI_B = & \left. \Big\{ \Phi^B_y: \cS(\cH) \to \cL^+_{\cH \otimes \cK} ~| \right. \\  \Phi^B_y(\rho) = & \left. \sum_x \sum_{z_{xy} \in P(x,y)} \ketbraq{ \phi_{z_{xy}} }( \rho \otimes \ketbraq{0} ) \ketbraq{ \phi_{z_{xy}} } \right. \\
= & \left. \sum_x \sum_{z_{xy} \in P(x,y)} K^B_{(x,z_{xy}),y} \rho (K^B_{(x,z_{xy}),y})^\dagger \Big\} \right.
\end{split}
\end{equation}
where $K^B_{(x,z_{xy}),y} = \ketbraq{ \phi_{z_{xy}} } (\Id_\cH \otimes \ket{0})$. It has been shown in Ref.~\cite{MF22} that $\cI_A$ and $\cI_B$ are compatible.

Using Proposition \ref{prop:detailed_equiv}, we show that $\cI_A$ and $\cI_B$ are post-processing equivalent to their detailed instruments, $\hat{\cI}_A$ and $\hat{\cI}_B$, respectively. Indeed, we have that
\begin{equation}
\begin{split}
& \left. (K^A_{(y,_{xy}z),x})^\dagger K^A_{(y',z'_{xy'}),x} \right. \\
& \left. = (\Id_\cH \otimes \bra{0}) \ket{ \phi_{z_{xy}} } \braket{ \phi_{z_{xy}} | \phi_{z'_{xy'}} } \bra{ \phi_{z'_{xy'}} } (\Id_\cH \otimes \ket{0}) = 0 \right.
\end{split}
\end{equation}
for all $x$ whenever $(y,z_{xy}) \neq (y',z'_{xy'})$, which implies that $\cI_A \sim \hat{\cI}_A$ by Proposition \ref{prop:detailed_equiv}. Similarly, we can show that $\cI_B \sim \hat{\cI}_B$. Therefore, following from Corollary \ref{cor:compatible_equiv}, we have that the indecomposable instruments $\hat{\cI}_A$ and $\hat{\cI}_B$ are compatible.
\end{proof}

\begin{remark}
While for every pair of compatible measurements there exists a pair of compatible instruments that induce the measurements, it is not true that all such instruments are compatible. Indeed, consider the L\"uders instrument $\Phi_x(\rho) = \sqrt{A(x)} \rho \sqrt{A(x)}$ of a measurement~$A$. Given two trivial measurements $A$ and $B$ (which are always compatible), the corresponding L\"uders instruments are not compatible, since they both induce the identity channel \cite{MF22}. However, following from Ref.~\cite{MF22}, there exists a pair of instruments compatible with $A$ and $B$ such that these instruments are compatible. Moreover, it was proven in Ref.~\cite{Hay06} that every $A$-compatible instrument is a post-processing of the L\"uders instrument of $A$. Thus, it would be interesting to characterize that starting from the (potentially incompatible) L\"uders instruments of a pair of compatible measurements $A$ and $B$, at what point in the post-processing preorder do we obtain a compatible pair of $A$ and $B$-compatible instruments. Another consequence of the fact that every $A$-compatible instrument is a post-processing of the L\"uders instrument of $A$ is that every instrument that is compatible with the L\"uders instrument of $A$ is compatible with all $A$-compatible instruments.
\end{remark}

\section{Conclusion}\label{Sec:conc}

In this paper, we have studied the quantification and characterization of the recently introduced concept of parallel compatibility of quantum instruments, which is a natural generalization of the more established notions of measurement and channel compatibility. We have defined the incompatibility robustness of quantum instruments to quantify the incompatibility of instruments, and have proved universal bounds on this quantity. Furthermore, we have proved that post-processing of quantum instruments is a free operation in a potential resource theory of quantum instruments. We also have proved that free PID supermaps, which are free operations for traditional compatibility, are not free operations for parallel compatibility. Then we have provided families of instruments for which our bounds are tight. Finally, we have proved that for every pair of compatible measurements there exists a pair of compatible indecomposable instruments.

Given this well-defined way of quantifying the incompatibility of quantum instruments, it is a natural further research direction to see how the incompatibility robustness relates to quantitative measures of the performance of instruments in quantum information processing tasks. Since quantum instruments appear in the description of sequential tasks, such as sequential prepare-and-measure protocols \cite{MTB19,MBP20} or sequential Bell scenarios \cite{GWCA+14,SGGP15,BC20}, it would be an interesting research direction to investigate the quantitative relationship between the incompatibility robustness of instruments and the success probability or Bell violation in these information processing tasks. Furthermore, while we provide some examples and counterexamples for free operations of parallel compatibility, it could be a potential further research direction to systematically study such free operations and develop a resource theory of instrument incompatibility.\\

\section{Acknowledgements}

The authors thank Prof.~Sibasish Ghosh and Prof.~Prabha Mandayam for their valuable comments.
This  project  has  received  funding  from  the  European  Union’s  Horizon~2020  research and innovation programme under the Marie Sk\l{}odowska-Curie grant agreement No.~754510. M.F.~acknowledges funding from the Government of Spain (FIS2020-TRANQI, Severo Ochoa CEX2019-000910-S), Fundaci\'o Cellex, Fundaci\'o Mir-Puig and Generalitat de Catalunya (CERCA, AGAUR SGR 1381).


\begin{thebibliography}{98}

\bibitem{HMZ16} T. Heinosaari, T. Miyadera, and M. Ziman, An invitation to quantum incompatibility, Journal of Physics A: Mathematical and Theoretical \textbf{49}, 123001 (2016).

\bibitem{Kim-Bell-incomp} W. Son, E. Andersson, S. M. Barnett, and M. S. Kim, Joint measurements and Bell inequalities, Phys. Rev. A \textbf{72}, 052116 (2005).

\bibitem{Vertesi-steer-incomp} M. Túlio Quintino, T. Vértesi, and N. Brunner, Joint Measurability, Einstein-Podolsky-Rosen Steering, and Bell Nonlocality, Phys. Rev. Lett. \textbf{113}, 160402 (2014).

\bibitem{Skrzypczyk-state-dis-incomp} P. Skrzypczyk, I. Šupić, and D. Cavalcanti, All sets of incompatible measurements give an advantage in quantum state discrimination, Phys. Rev. Lett. \textbf{122}, 130403 (2019).

\bibitem{Heino-QRAC} C. Carmeli, T. Heinosaari, and A. Toigo, Quantum random access codes and incompatibility of measurements, Europhys. Lett. \textbf{130}, 50001 (2020).

\bibitem{WPF09} M. M. Wolf, D. Perez-Garcia, and C. Fernandez, Measurements incompatible in quantum theory cannot be measured jointly in any other no-signaling theory, Phys. Rev. Lett., {\bf 103} 230402 (2009).

\bibitem{UBGP15} R. Uola, C. Budroni, O. Gühne, and J.-P. Pellonpää, One-to-one mapping between steering and joint measurability problems, Phys. Rev. Lett., {\bf 115}, 230402 (2015).

\bibitem{TU20} A. Tavakoli, and R. Uola, Measurement incompatibility and steering are necessary and sufficient for operational contextuality, Phys. Rev. Research {\bf 2}, 013011 (2020).

\bibitem{MF22} A. Mitra and M. Farkas, Compatibility of quantum instruments, Phys. Rev. A \textbf{105}, 052202 (2022).

\bibitem{Heino-book} T. Heinosaari and M. Ziman, \emph{The Mathematical Language of Quantum Theory: From Uncertainty to Entanglement} (Cambridge University Press, Cambridge, UK, 2012).

\bibitem{Buscemi-res-incomp} F. Buscemi, E. Chitambar, and W. Zhou, Complete Resource Theory of Quantum Incompatibility as Quantum Programmability, Phys. Rev. Lett. \textbf{124}, 120401 (2020).

\bibitem{Heino-incom-chan}  T. Heinosaari and T. Miyadera, Incompatibility of quantum channels, J. Phys. A: Math. Theor. \textbf{50}, 135302 (2017).

\bibitem{Plavala-incomp-chan} M. Girard, M. Plávala, and J. Sikora, Jordan Products of Quantum Channels and Their Compatibility, Nat. Commun. \textbf{12}, 2129 (2021).

\bibitem{Plavala-incomp-chan-gen-prob} M. Plávala, Conditions for the compatibility of channels in general probabilistic theory and their connection to steering and Bell nonlocality, Phys. Rev. A \textbf{96}, 052127 (2017).

\bibitem{Mori-incomp-chan} J. Mori, Operational characterization of incompatibility of quantum channels with quantum state discrimination, Phys. Rev. A \textbf{101}, 032331 (2020).


\bibitem{Davies-in} E. B. Davies, and J. T. Lewis, An operational approach to quantum probability, Commun. Math. Phys. \textbf{17}, 239 (1970).

\bibitem{Pellonpaa-1} J. P. Pellonpää, Quantum instruments: I. Extreme instruments,  \href{https://doi.org/10.1088/1751-8113/46/2/025302}{J. Phys. A: Math. Theor. \textbf{46} 025302 (2013)}.

\bibitem{Pellonpaa-2} J. P. Pellonpää, Quantum instruments: II. Measurement theory, \href{https://doi.org/10.1088/1751-8113/46/2/025303}{J. Phys. A: Math. Theor. \textbf{46} 025303 (2013)}.

\bibitem{HMR14} T. Heinosaari, T. Miyadera, and D. Reitzner,  Strongly incompatible quantum devices, Found. Phys. \textbf{44}, 34 (2014). 

\bibitem{Lev16} F. Lever, \emph{Measurement incompatibility: a resource for quantum steering, Master's thesis}, \href{https://www.physik.uni-siegen.de/tqo/publications/theses/fabiano-lever-master-thesis.pdf}{University of Trento (2016)}.

\bibitem{Gud20} S. Gudder, Finite quantum instruments (2020).

\bibitem{LS21} L. Leppäjärvi and M. Sedlák, Postprocessing of quantum instruments, Phys. Rev. A \textbf{103}, 022615 (2021).

\bibitem{DFK19} S. Designolle, M. Farkas, and J. Kaniewski, Incompatibility robustness of quantum measurements: A unified framework, New J. Phys. \textbf{21}, 113053 (2019).

\bibitem{BV04} S. Boyd and L. Vandenberghe, Convex Optimization. New York, NY, USA: Cambridge University Press, 2004.

\bibitem{JC21} K. Ji and E. Chitambar, Incompatibility as a resource for programmable quantum instruments, 	arXiv:2112.03717 [quant-ph] (2021).

\bibitem{Hay06} M. Hayashi, Quantum Information: An Introduction (SpringerVerlag, 2006) [translated from the 2003 Japanese original].

\bibitem{MTB19} K. Mohan, A. Tavakoli, N. Brunner, Sequential random access codes and self-testing of quantum measurement instruments, New J. Phys. {\bf 21}, 083034 (2019)

\bibitem{MBP20} N. Miklin, J. J. Borkała, and M. Pawłowski, Semi-device-independent self-testing of unsharp measurements, Phys. Rev. Research {\bf 2}, 033014 (2020)

\bibitem{GWCA+14} R. Gallego, L. E. W\"urflinger, R. Chaves, A. Ac\'in, and M. Navascu\'es, Nonlocality in sequential correlation scenarios, New Journal of Physics {\bf 16}, 033037 (2014)

\bibitem{SGGP15} R. Silva, N. Gisin, Y. Guryanova, and S. Popescu, Multiple observers can share the nonlocality of half of an entangled pair by using optimal weak measurements, Phys. Rev. Lett. {\bf 114}, 250401 (2015)

\bibitem{BC20} P. J. Brown and R. Colbeck, Arbitrarily many independent observers can share the nonlocality of a single maximally entangled qubit pair, Phys. Rev. Lett. {\bf 125}, 090401 (2020)

 
\end{thebibliography}
\end{document}